\documentclass{revtex4}
\usepackage{amsmath,amsthm}
\usepackage{setspace}
\usepackage{graphicx}
\newcommand{\bra}[1]{\langle#1|}
\newcommand{\ket}[1]{|#1\rangle}
\usepackage{amsfonts}
\newtheorem{theorem}{Theorem}[section]
\newtheorem{lemma}[theorem]{Lemma}
\newtheorem{proposition}[theorem]{Proposition}

\theoremstyle{remark}
\newtheorem{remark}[theorem]{Remark}

\theoremstyle{definition}

\theoremstyle{example}

\theoremstyle{notation}

\begin{document}
\title{Lower and upper probabilities in the distributive lattice of subsystems}            
\author{A. Vourdas}
\affiliation{Department of Computing,\\
University of Bradford, \\
Bradford BD7 1DP, United Kingdom\\
{\rm a.vourdas@bradford.ac.uk}\\
Tel.:44-1274-233950}

\begin{abstract}
The set of subsystems $\Sigma (m)$ of a finite quantum system $\Sigma (n)$ (with variables in ${\mathbb Z}(n)$)
together with logical connectives, is a distributive lattice. With regard to this lattice, the 
$\ell (m|\rho_n)={\rm Tr}[{\mathfrak P}(m)\rho _n]$ (where ${\mathfrak P}(m)$ is the projector to $\Sigma (m)$)
obeys a supermodularity inequality, and it is interpreted as a lower probability 
in the sense of the Dempster-Shafer theory, and not as a Kolmogorov probability. It is shown that the basic concepts of the Dempster-Shafer theory 
(lower and upper probabilities and the Dempster multivaluedness) are pertinent to the quantum formalism of finite systems.

\end{abstract}
\maketitle

\section{Introduction}

When we have a structure (e.g., groups in algebra), we also introduce substructures (subgroups) and study the relationship between them.
This is our motivation for the study of subsystems of a finite quantum system $\Sigma (n)$, with variables in ${\mathbb Z}(n)$.
A subsystem of $\Sigma (n)$ is a system with variables in a subgroup of ${\mathbb Z}(n)$.
The subgroups of ${\mathbb Z}(n)$ are the ${\mathbb Z}(m)$ with $m|n$, and therefore the subsystems of
$\Sigma (n)$ are the $\Sigma (m)$ with $m|n$. 
The various subsystems $\Sigma (m)$ are embedded into $\Sigma (n)$ as described below. The projectors 
${\mathfrak P}(m)$ into the subsystems $\Sigma (m)$ commute with each other, and they can be associated 
with commuting measurements.
The set of subsystems of $\Sigma (n)$ (with logical connectives) is a distributive lattice $\Lambda ({\bf \Sigma}_n)$.

Our concept `subsystem' is linked to subgroups, because
the positions in a subsystem take values in a subgroup of the group of positions of the full system.
The Hilbert space $H(m)$ of $\Sigma (m)$ is a subspace of the Hilbert space $H(n)$ of $\Sigma (n)$, 
but there is no subsystem for every subspace of $H(n)$.
Our concept `subsystem' is much stronger than the concept `subspace'.
The lattice of the subgroups of ${\mathbb Z}(n)$ is distributive, and consequently the lattice $\Lambda ({\bf \Sigma}_n)$ of subsystems of $\Sigma (n)$ is distributive.
In fact it is a Heyting algebra, because every finite distributive lattice is a Heyting algebra\cite{BIR,BIR1,BIR2,BIR3}.

Probability theory is tacitly defined with respect to a lattice, because in its axioms it needs the concepts of conjuction, disjunction and negation.
Kolmogorov probability is defined on a powerset $2^\Omega$, which is a Boolean algebra, and where
the intersection, union and complement play the role of conjunction, disjunction and negation.
Quantum logic studies the orthomodular lattice of closed subspaces of a Hilbert space \cite{LO1,LO2,LO3,LO4}, which
has various Boolean algebras as sublattices, and Kolmogorov probabilities are defined on them.

In ref\cite{vou55,vou56} we have studied a different lattice which is the distributive lattice (Heyting algebra)
of the subsystems of a finite quantum system.
We have shown that the quantum probabilities, 
$\ell (m|\rho_n)={\rm Tr}[{\mathfrak P}(m)\rho _n]$ (where $\rho _n$ is a density matrix of the system $\Sigma (n)$),
 obey the supermodularity inequality 
\begin{eqnarray}\label{pro1}
\ell(m_1\vee m_2|\rho _n)-\ell (m_1|\rho _n)-\ell (m_2|\rho _n)
+\ell(m_1\wedge m_2|\rho _n)\ge 0.
\end{eqnarray}
In contrast, Kolmogorov probabilities $q(m)$ obey the modularity equality
\begin{eqnarray}\label{pro2}
q(m_1\vee m_2)-q(m_1)-q(m_2)+q(m_1\wedge m_2)=0,
\end{eqnarray}
Only in the special case that the variables $m_1,m_2$ belong to the same chain, Eq(\ref{pro1}) is valid as an equality.
Within a chain the quantum probabilities $\ell (m|\rho_n)$ obey an equality analogous to Eq.(\ref{pro2}), i.e., they behave like
Kolmogorov probabilities.

There are many problems in Artificial Intelligence,
Operations Research, Economics, etc,  which use probabilities with the property of Eq.(\ref{pro1}). In these subjects
we have conflicting data, and concepts like imprecise probability and non-additive probability \cite{F1,F2,F3,F4,F5}, 
have been introduced 
in order to reconcile the contradictions.
Among such theories, the  Dempster-Shafer approach \cite{DS1,DS2,DS3,DS4,DS5,DS6} has been used extensively in Artificial Intelligence, 
and in this paper we use it in the formalism of finite quantum systems. The Dempster-Shafer theory 
fits very well with the
fact that the $\ell (m|\rho_n)$ obey Eq.(\ref{pro1}), but not Eq.(\ref{pro2}).
The terminology used in quantum mechanics is sometimes different from the terminology  used in the Dempster-Shafer theory, and we
provide a `translation' between the two.

In section 2, we discuss briefly, submodular and supermodular functions,
the lattice structure of the set of subgroups of ${\mathbb Z(n)}$, and capacities (a concept weaker than probability measures), in order to establish 
the notation.
In section 3 we present some aspects of the Dempster-Shafer theory (lower and upper probabilities, multivaluedness, etc) which are used later.
In section 4, we provide a deeper insight to the fact that, 
with respect to the distributive lattice $\Lambda ({\bf \Sigma }_n)$ of the subsystems of $\Sigma (n)$, the quantum probabilities $\ell (m|\rho_n)$ obey 
the inequality of Eq.(\ref{pro1}), but they do not obey the equality of Eq.(\ref{pro2}).
In section 5, we show that the properties of the quantum probabilities $\ell (m|\rho_n)$ fit very well with the Dempster-Shafer theory.
We conclude in section 6, with a discussion of our results.

\section{Preliminaries}

\begin{itemize}

\item[(1)]
$r|s$  or $r\prec s$ denotes that $r$ is a divisor of $s$.
${\rm GCD}(r,s)$ and ${\rm LCM}(r,s)$ are the greatest common divisor and least common multiplier correspondingly, of the integers $r,s$.

${\mathbb D}(n)$ is the set of divisors of $n$.
The cardinality of ${\mathbb D}(n)$ is given by the divisor function $\sigma _0(n)$.
A divisor $r$ of $n$, such that $r$ and $n/r$ are coprime, is called a Hall divisor of $n$ (a terminology inspired by group theory).

\item[(2)]
${\mathbb Z}(n)$ is the ring of integers modulo $n$.
If $m\prec n$ then  ${\mathbb Z}(m)$ is a subgroup of ${\mathbb Z}(n)$. 
${\mathbb Z}^*(n)$ is the reduced system of residues modulo $n$. It contains the invertible elements 
of ${\mathbb Z}(n)$.

\item[(3)]
A set $A$ viewed as a lattice (i.e., with the operations $\vee$ and $\wedge$) is denoted as $\Lambda (A)$.
Throughout the paper we have various lattices and for simplicity we use the same symbols $\prec$, $\wedge $, $\vee$, $\neg$,  for the `partial order',
`meet', `join' and `negation', correspondingly. 
We also use the same symbols  $\cal O$ and $\cal I$ for the smallest and greatest elements.

All our lattices are finite distributive lattices. As such they are Heyting algebras and obey the relations $a \prec \neg \neg a$ and
$a\vee \neg a\prec {\cal I}$. A Heyting algebra may have a sublattice which is a Boolean algebra,
and for its elements $\neg \neg a=a$ and $a\vee \neg a={\cal I}$.
The $a\vee \neg a ={\cal I}$ is the `law of the excluded middle',  and it is is not valid in Heyting algebras, but it is valid in Boolean algebras.
The formalism of subsystems of $\Sigma (n)$ as a Heyting algebra, and the physical meaning of the logical connectives, is discussed in \cite{vou55}. 
Here we only need a minimal amount of these ideas.

\item[(4)]
Kolmogorov probability theory is defined on the powerset $2^\Omega$ of a set $\Omega$. This is a Boolean algebra 
which has the $\subseteq$, $\cap$ and $\cup$, as the logical connectives $\prec$, $\wedge $, $\vee$, correspondingly. 
The negation $\neg A$ of a subset of $\Omega$, is the complement ${\overline A}=\Omega -A$.

\item[(5)]
Sometimes in the literature, in a bipartite system described with the tensor product $H_A\otimes H_B$, 
the term subsystem is used for each of the two parties indexed with $A,B$.
Here the term subsystem means something different.
If $\Sigma (n)$ is a system with variables in ${\mathbb Z}(n)$,
a subsystem $\Sigma (m)$ of $\Sigma (n)$, is a system with variables in a subgroup of ${\mathbb Z}(n)$.
And there is an embedding of the subsystem $\Sigma (m)$ into $\Sigma (n)$, which is described explicitly below.
\end{itemize}

\subsection{Submodular and supermodular functions}

Let $f(m)$ be a function on a distributive lattice $\Lambda$ and 
\begin{eqnarray}\label{h1}
F(m_1,m_2)=f(m_1\vee m_2)-f(m_1)-f(m_2)+f(m_1\wedge m_2);\;\;\;\;m_i\in {\Lambda}.
\end{eqnarray}
$f(m)$ is supermodular, modular or submodular, if for all $m_1,m_2$, we get
$F(m_1,m_2)\ge 0$, $F(m_1,m_2)= 0$, $F(m_1,m_2)\le 0$, correspondingly.
We note that if $m_1\prec m_2$ or $m_1=\cal O$ or $m_1=\cal I$, then $F(m_1,m_2)=0$.
Supermodular and submodular functions have been studied and used in Optimization and Operations Research \cite{F}. 
Their properties are related to the fact that $F(m_1,m_2)$ can be viewed as a `discretized second derivative'.

For functions such that $f({\cal O})=0$, and for $m_1\wedge m_2={\cal O}$, Eq.(\ref{h1}) gives 
\begin{eqnarray}\label{h2}
F(m_1,m_2)=f(m_1\vee m_2)-f(m_1)-f(m_2), 
\end{eqnarray}
and supermodularity, modularity or submodularity, reduces to superadditivity, additivity or subadditivity, correspondingly.

\subsection{The lattice $\Lambda [{\mathbb D}(n)]$ of divisors of $n$}

We factorize the number $n$ in terms of prime numbers as
\begin{eqnarray}
n=\prod _{p\in \varpi (n)}p^{e_p(n)}
\end{eqnarray}
Here $\varpi (n)$ is the set of prime numbers in this factorization, and $e_p(n)$ is the exponent of $p$.
If $k\in {\mathbb D}(n)$ then 
\begin{eqnarray}
k=\prod _{p\in \varpi (k)}p^{e_p(k)};\;\;\;\;\;\varpi (k)\subset \varpi (n);\;\;\;\;\;e_p(k)\le e_p(n)
\end{eqnarray}
The set ${\mathbb D}(n)$ with divisibility as partial order, and with
\begin{eqnarray}
k\wedge m={\rm GCD}(k,m);\;\;\;\;\;k\vee m={\rm LCM}(k,m);\;\;\;\;\;\neg k=\prod _{p\in \varpi (n)-\varpi(k)}p^{e_p(n)}
\end{eqnarray}
is a finite distributive lattice and as such it is a Heyting algebra with ${\cal O}=1$ and ${\cal I}=n$. 
We denote it as $\Lambda [{\mathbb D}(n)]$.
$\neg k$ is the largest element of ${\mathbb D}(n)$ such that $k\wedge (\neg k)=1$.
The following subset of ${\mathbb D}(n)$
\begin{eqnarray}
{\mathbb D}^B(n)=\left \{\prod _{p\in \pi}p^{e_p(n)}\;|\;\pi \subseteq \varpi(n)\right \},
\end{eqnarray}
contains the Hall divisors of $n$, and it is a Boolean algebra. If all exponents $e_p(n)=1$, then ${\mathbb D}^B(n)={\mathbb D}(n)$.

\subsection{The lattice $\Lambda [{\mathfrak Z}(n)]$ of subgroups of ${\mathbb Z}(n)$}\label{AAS}

We consider the set 
\begin{eqnarray}
{\mathfrak Z}(n)=\{{\mathbb Z}(m)\;|\;m\in {\mathbb D}(n)\};\;\;\;\;\;n\in{\mathbb N},
\end{eqnarray}
which contains the subgroups of ${\mathbb Z}(n)$.
A subgroup ${\mathbb Z}(m)$ can be embedded into a larger group ${\mathbb Z}(k)$ (where $m\prec k \prec n$), with the map
\begin{eqnarray}\label{b5}
{\mathbb Z}(m)\ni a\;\;\rightarrow\;\;da \in {\mathbb Z}(k);\;\;\;\;d=\frac{k}{m}
\end{eqnarray} 

The ${\mathfrak Z}(n)$ with subgroup as partial order, and with
\begin{eqnarray}
{\mathbb Z}(k\wedge m)={\mathbb Z}(k)\wedge {\mathbb Z}(m);\;\;\;\;\;{\mathbb Z}(k\vee m)={\mathbb Z}(k)\vee {\mathbb Z}(m)
;\;\;\;\;\;\neg {\mathbb Z}(k)={\mathbb Z}(\neg k)
\end{eqnarray}
is a Heyting algebra with 
${\cal O}={\mathbb Z}(1)$ and ${\cal I}={\mathbb Z}(n)$.
It is isomorphic to $\Lambda [{\mathbb D}(n)]$ and we denote it as $\Lambda [{\mathfrak Z}(n)]$.

${\mathbb Z}(k\vee m)$ is the smallest group in ${\mathfrak Z}(n)$ which has the ${\mathbb Z}(k)$ and ${\mathbb Z}(m)$
as subgroups. Taking into account the map of Eq.(\ref{b5}), we see that ${\mathbb Z}(k\vee m)$ contains all the elements of both 
${\mathbb Z}(k)$ and ${\mathbb Z}(m)$, and also the elements of ${\mathbb Z}^*(k\vee m)$, which
as explained later, cause the supermodularity of $\ell (m|\rho _n)$, in Eq.(\ref{pro1}).
${\mathbb Z}(k\wedge m)$ is the largest subgroup of both ${\mathbb Z}(k)$ and ${\mathbb Z}(m)$.
$\neg {\mathbb Z}(k)$ is the largest group in ${\mathfrak Z}(n)$ such that $[\neg {\mathbb Z}(k)]\wedge {\mathbb Z}(k)={\mathbb Z}(1)$.

The subset of $\Lambda [{\mathfrak Z}(n)]$, given by
\begin{eqnarray}
\Lambda [{\mathfrak Z}^B(n)]=\left \{{\mathbb Z}\left (\prod _{p\in \pi}p^{e_p(n)}\right )\;|\;\pi \subseteq \varpi(n)\right \},
\end{eqnarray}
is a Boolean algebra.

\subsection{ Capacities or non-additive probabilities}

Sometimes there is added value in a coalition (e.g., in the merger of two companies).
In everyday language this is described with the expression `one plus one is three', or the expression
`the whole is greater than the sum of its parts'. 
Of course, the added value can be negative.
In such cases probability is not additive. The term 
capacity is used for non-additive probabilities (e.g., \cite{CC1,CC2}).

Let $2^\Omega$ be the powerset of a set $\Omega$, which in our case is finite.
A capacity or non-additive probability in $\Omega$, is a function $\mu$ from $2^\Omega$ to $[0,1]$, such that
\begin{eqnarray}
&&\mu (\emptyset )=0;\;\;\;\;\;\mu (\Omega) =1\\
&&A\subseteq B \subseteq \Omega \;\;\rightarrow \mu (A)\le \mu (B)\label{BB}
\end{eqnarray}
If we replace the monotonicity relation of Eq.(\ref{BB}) with the additivity property 
\begin{eqnarray}
&&A\cap B=\emptyset\;\;\rightarrow \mu (A\cup B)=\mu(A)+\mu (B)\label{AA}
\end{eqnarray}
which is stronger, we get a probability measure.

Let $\Omega =\{1,...,n\}$. The non-additivity of capacities, i.e., the fact that in general $\mu(\{i,j\})-\mu (\{j\})\ne \mu (\{i\})$ 
implies that $\mu (\{i\})$ is not a good estimate of the weight (or importance) of the element $i$, in the formalism.
The 
\begin{eqnarray}\label{C}
{\mathfrak W}(i|B)=\mu (B\cup \{i\})-\mu (B)
\end{eqnarray}
is the added value to the weight (or importance) of the element $i$ when it is in a coalition with the subset $B\subseteq \Omega$.
We can introduce a weighted average of these quantities as a measure of the overall importance of the element $i$.
This can be done in various ways.  Shapley \cite{CC2,S} introduced one of them, which is used in game theory, 
as a measure of the importance of each player within a coalition. We do not pursue further this direction.

\section{Multivaluedness and the Dempster-Shafer theory}\label{A}

Let $X$ be a sample space, and $\Gamma$ a multivalued map from $X$ to another sample space $\Omega$.
The Dempster-Shafer theory carries probabilities on subsets of $X$, into lower and upper probabilities on subsets of $\Omega$.
The need for lower and upper probabilities arises from the fact that $\Gamma$ is not single-valued.
In this case we have an ambiguity in the probability, which is expressed with the interval from the lower to the upper probability.
We first give an example, and then summarize the properties of the upper and lower probabilities, and compare them with those of Kolmogorov probabilities.
The analogues of these ideas for the quantum formalism are discussed explicitly, later.

\subsection{Uncertainty and ambiguity: an example}\label{exa}

A class has $n$ students $S_1,...,S_n$. An attribute for each student $S_i$ is not accurately known, but it takes
values in a set $G_i$ of integer numbers. 
For example, each student does a `final year project', and each project is assessed independently by many professors.
The set $G_i$ contains all the marks (integers in the interval $[0,100]$)
for the project of the student $S_i$.
The number of professors assessing each project may vary from one student to another, i.e., the various sets $G_i$ have different cardinalities, in general.
We have a multivalued map, where to each student corresponds a set of marks. 
We calculate the probability that a random student has marks within a given set $A$ (e.g., above $70$).

There are three categories of students. 
\begin{itemize}
\item
The first category contains
$n_1$ students such that $G_i\subseteq A$. 
For these students we are certain that their marks belong to the set $A$. 
\item
The second category contains $n_2$ students such that 
$G_i\cap A\ne \emptyset$ and also $G_i\cap {\overline A}\ne \emptyset$. 
For these students, some (but not all) of their marks belong to $A$.
Dempster \cite{DS2} uses the term `don't know' for this category.

\item
The third category contains the $n-n_1-n_2$ students such that $G_i\subseteq {\overline A}$. 
For these students we are certain that the marks do not belong to the set $A$. 
\end{itemize}
The $n_1+n_2$ students in the first two categories, can be described collectively by saying it is not true that $G_i\subseteq {\overline A}$. So the $G_i\subseteq A$ is not equivalent to the negation of $G_i\subseteq {\overline A}$.  

The lower probability or belief is $\ell(A)=n_1/n$, and is associated with the statement $G_i\subseteq A$.
The upper probability or plausibility is $u(A)=(n_1+n_2)/n$, and is associated with the negation of the statement $G_i\subseteq {\overline A}$.
The lower (upper) probability, simply excludes (includes) all the `don't know' cases.

Table \ref{360}, shows the marks for the projects of four students (ideally we should have an example with large $n$, but for practical reasons we take $n=4$). Table \ref{361}, shows 
the lower and upper probabilities $\ell(A_i)$ and $u(A_i)$ for the sets
\begin{eqnarray}\label{79}
&&A_1=\{m\;|\;60\le m\le 69\};\;\;\;\;A_2=\{m\;|\;70\le m\le 100\}\nonumber\\
&&A_1\cup A_2=\{m\;|\;60\le m\le 100\};\;\;\;\;\;\;A_3=\{m\;|\;65\le m\le 75\}.
\end{eqnarray}

There is much discussion in the literature about the normalization of the probabilities, in the case that some of the sets $G_i$ are empty,
i.e., some students have not been assessed. 
We do not consider this case, i.e., we assume that all sets $G_i$ with $i=1,...,n$, are non-empty.  
The above example is similar to the one in ref.\cite{DS6}, where a company does not know the age of its employees.
Several of its managers provide an estimate for the age of each employee, and this corresponds to the grades for each project, in our example.

For each student $S_i$ we choose one mark $a_i\in G_i$, and then we have the single-valued map 
which we denote as $\Gamma_{\nu}$. There are clearly many such maps indexed with $\nu$, and the mark of 
each student depends on the choice of $\Gamma_{\nu}$.
With the single-valued map $\Gamma_{\nu}$, let $k_{\nu}$ be the number of students with mark in the set $A$.
Then the probability 
that a random student has a mark within the set $A$, is simply $q_{\nu}(A)=k_{\nu}/n$. For any $\Gamma_{\nu}$, we get
$n_1\le k_{\nu}\le n_2$ and therefore $\ell (A)\le q_{\nu}(A) \le u(A)$.

There are two different kinds of indeterminateness in these examples.
The first is associated to probabilities $q_{\nu}(A)$ with fixed $\nu$.
The second is related to the fact that we have many $\nu$, and therefore many $q_{\nu}(A)$ for a fixed set $A$.
In order to distinguish them, we need two different terms, and following ref\cite{F2}, we call the former 
uncertainty and the latter ambiguity. Ambiguity is intimately related to the multivaluedness, and it refers to the fact that we have
an interval $[\ell(A),u(A)]$ of probabilities, rather than a single probability.
In the case of singlevaluedness (i.e., if we have a single grade for each project), $\ell(A)=u(A)$, and there is no ambiguity.

\subsection{Properties of lower and upper probabilities}\label{lu}

Let $A,B$ be elements of the powerset $2^{\Omega}$ (i.e., subsets of $\Omega$).
Kolmogorov's probability obeys the relations
\begin{eqnarray}
&&q (\emptyset)=0;\;\;\;\;q (\Omega)=1\\
&&q (A\cup B)-q (A)-q (B)+q (A\cap B)= 0,\label{76b}
\end{eqnarray}
and from this follows that 
\begin{eqnarray}
&&A \subseteq B\;\;\rightarrow\;\;q(A)\le q (B)\\
&&q(A)+q({\overline A})=1\label{21}.
\end{eqnarray}

The lower probability or belief $\ell (A)$, is a monotone function from $2^{\Omega}$ to $[0,1]$, i.e., 
\begin{eqnarray}\label{28A}
A \subseteq B\;\;\rightarrow\;\;\ell (A)\le \ell (B),
\end{eqnarray}
and it obeys the relations:
\begin{eqnarray}
&&\ell (\emptyset)=0;\;\;\;\;\ell (\Omega)=1\label{74}\\
&&\ell (A\cup B)-\ell (A)-\ell (B)+\ell (A\cap B)\ge0.\label{76}
\end{eqnarray}
From this follows that 
\begin{eqnarray}\label{1z}
\ell (\overline A)+\ell (A)\le 1
\end{eqnarray}
For Kolmogorov probabilities $1-q({\overline A})=q(A)$, but for lower probabilities
the $1-\ell (\overline A)$ is a different from $\ell (A)$, and we call it upper probability or plausibility $u(A)$:
\begin{eqnarray}\label{35}
u(A)=1-\ell (\overline A)\ge \ell (A)
\end{eqnarray}
Using the equations above, we prove that
\begin{eqnarray}
&&u (\emptyset)=0;\;\;\;\;u (\Omega)=1\\
&&A \subseteq B\;\;\rightarrow\;\;u (A)\le u (B)\label{34A}\\
&&u (A\cup B)-u (A)-u (B)+u (A\cap B)\le 0.\label{76a}
\end{eqnarray}
From Eqs (\ref{76}), (\ref{76a}) with $A\cap B=\emptyset$, it follows that both the lower probability and the upper probability are capacities.

The difference between the upper and lower probabilities, describes the `don't know' case:
\begin{eqnarray}
d(A)=u(A)-\ell (A)=1-\ell (A)-{ \ell } ( \overline A).
\end{eqnarray}
The upper probability combines the `true' and the `don't know'. 
Ref.\cite{DS2} discusses the importance of the `don't know' category.

\begin{remark}
In this section we have defined lower and upper probabilities on a powerset $2^\Omega$, which is a Boolean algebra.
Therefore ${\overline {\overline A}}=A$, which implies that $\ell ({\overline {\overline A}})=\ell (A)$ and $u ({\overline {\overline A}})=u (A)$. 
Below we will define lower and upper probabilities on a Heyting algebra, where $a \prec \neg\neg a$.
\end{remark}

\section{Subsystems of $\Sigma(n)$}

\subsection{Embedding of subsystems}

$\Sigma (n)$ is a quantum system with positions and momenta in ${\mathbb Z}(n)$, and $n$-dimensional
Hilbert space $H(n)$.
$|X_n;r\rangle$ where $r\in {\mathbb Z}(n)$, is an orthonormal basis that we call `basis of position states'
(the $X_n$ in this notation is not a variable, but it simply indicates that they are position states).
Through a Fourier transform we get another orthonormal basis that we call momentum states\cite{FI1}:
\begin{equation}
|{P_n};r\rangle=F_n|{X_n};r\rangle;\;\;\;\;
F_n=n^{-1/2}\sum _{r,s}\omega _n(rs)\ket{X_n;r}\bra{X_n;s};\;\;\;\;
\omega_n(r)=\exp \left (i\frac{2\pi r}{n}\right )
\end{equation}
For $m\prec k\prec n$, the $\Sigma (m)$ is a subsystem of $\Sigma (k)$
(which we denote as $\Sigma (m)\prec \Sigma (k)$), and
the space $H(m)$ is a subspace of $H(k)$ (which we denote as $H(m)\prec H(k)$).
We can embed the states of  $\Sigma (m)$ into $\Sigma (k)$, as follows:
\begin{eqnarray}\label{56}
{\cal A}_{mk}:\;\;\sum _{r=0}^{m-1}a_r\ket{X_m;r}\;\;\rightarrow\;\;\sum _{r=0}^{m-1}a_r\ket{X_k;\frac{kr}{m}};\;\;\;\;m\prec k.
\end{eqnarray}
The system $\Sigma (1)$ is 
physically trivial, as it has one-dimensional Hilbert space $H(1)$ which consists of the `vacuum' state $\ket{X_1;0}=\ket{P_1;0}$.

We define the projector
\begin{eqnarray}\label{oper}
{\mathfrak P}(m)=\sum _{r=0}^{m-1}\ket {X_{k};\frac{kr}{m}}\bra{X_{k};\frac{kr}{m}};\;\;\;\;m\prec k;\;\;\;\;m,k\in {\mathbb D}(n).
\end{eqnarray}
The map of Eq.(\ref{56}), which identifies the state $\ket {X_{m};r}$ in $H(m)$ with the state $\ket {X_{k};\frac{kr}{m}}$ in $H(k)$, 
implies that we do not need to use an index $k$ to denote this projector as ${\mathfrak P}_k(m)$.
Also ${\Sigma }(n)$ is the largest system, and therefore
${\mathfrak P}(n)={\bf 1}_n$. All these projectors commute with each other.

\subsection{The lattice $\Lambda ({\bf \Sigma}_n )$ of subsystems}

Let ${\bf \Sigma}_n$ be the set of subsystems of $\Sigma (n)$ and ${\bf H}_n$ the set of their Hilbert spaces:
\begin{eqnarray}
&&{\bf \Sigma}_n=\{\Sigma (m)\;|\;m\in {\mathbb D}(n)\}\nonumber\\
&&{\bf H}_n=\{H(m)\;|\;m\in {\mathbb D}(n)\}.
\end{eqnarray}
${\bf \Sigma}_n$ is a partially ordered set with partial order `subsystem'.
${\bf H}_n$ is a partially ordered set with partial order `subspace'.

The set ${\bf \Sigma}_n $  with 
\begin{eqnarray}\label{123}
&&{\Sigma}(m)\wedge {\Sigma}(k)=
{\Sigma}(m \wedge k)\nonumber\\
&&\Sigma (m)\vee \Sigma (k)=\Sigma (m\vee k)\nonumber\\
&&\neg \Sigma (m)=\Sigma (\neg m);\;\;\;\;m,k\in {\mathbb D}(n)
\end{eqnarray}
is a distributive lattice (Heyting algebra) with ${\cal O}=\Sigma (1)$ and ${\cal I}=\Sigma (n)$.
It is isomorphic to $\Lambda [{\mathbb D}(n)]$ and we denote it as $\Lambda({\bf \Sigma}_n)$.
The physical meaning of the connectives is (see also ref\cite{vou55})
\begin{itemize}
\item
$\Sigma (m)\vee \Sigma (k)$ is the smallest subsystem that contains both $\Sigma (m)$ and $\Sigma (k)$, and in this sense 
it is the `coalition' or `merger' of these subsystems (they are terms used in the literature on the Dempster-Shafer theory).

\item
$\Sigma (m)\wedge \Sigma (k)$ is the largest common subsystem of  $\Sigma (m)$ and $\Sigma (k)$.

\item
$\neg \Sigma (m)$ is the maximal subsystem in $\Lambda({\bf \Sigma}_n)$ such that
$[\neg \Sigma (m)]\wedge \Sigma (m)=\Sigma (1)$. 
The $\neg \Sigma (m)$ and $\Sigma (m)$ share only the lowest state
$\ket{X_1;0}$. 

\end{itemize}

In analogous way we define the logical operations in ${\bf H}_n$, which 
is a Heyting algebra isomorphic to 
$\Lambda [{\mathbb D}(n)]$ and $\Lambda ({\bf \Sigma}_n)$, and we denote it as $\Lambda ({\bf H}_n)$.

All logical operations are linked to commuting von Neumann measurements.
For $m,k\in {\mathbb D}(n)$,
the ${\mathfrak P}(m\vee k)$ and ${\mathfrak P}(m\wedge k)$ are projectors to the spaces of 
the systems ${\Sigma}(m)\vee {\Sigma}(k)$
and ${\Sigma}(m)\wedge {\Sigma}(k)$, correspondingly.
Starting from a state of ${\Sigma}(n)$,
with these projectors we can get states in ${\Sigma}(m)\vee {\Sigma}(k)$
and ${\Sigma}(m)\wedge {\Sigma}(k)$.
Also the ${\mathfrak P}(\neg m)$ is the projector
to the space of the systems$\neg {\Sigma}(m)$.

The following proposition is a summary of results proved in \cite{vou55} and we give it without proof:
\begin{proposition}
For variables in ${\mathbb D}(n)$:
\begin{itemize}
\item[(1)]
\begin{eqnarray}\label{61}
&&H(m\wedge k)=H(m)\cap H(k)\nonumber\\
&&{\mathfrak P}(m\wedge k)={\mathfrak P}(m){\mathfrak P}(k)\nonumber\\
&&{\mathfrak P}(m){\mathfrak P}(\neg m)={\mathfrak P}(1)
\end{eqnarray}
\item[(2)]
\begin{eqnarray}\label{1000}
H(m_1\vee m_2)=T(m_1,m_2)\oplus S(m_1,m_2).
\end{eqnarray}
The space $T(m_1,m_2)$ 
\begin{eqnarray}\label{e1}
T(m_1,m_2)={\rm span}[H(m_1)\cup  H(m_2)]
\end{eqnarray}
contains all superpositions of states in $H(m_1)$ and $H(m_2)$, and it is a subspace of the space $H(m_1\vee m_2)$. 
The space $S(m_1,m_2)$ is orthogonal to the space $T(m_1,m_2)$ and contains disjunctions which are not superpositions.
The 
\begin{eqnarray}\label{r1}
{\mathfrak T}(m_1,m_2)={\mathfrak P}(m_1)+{\mathfrak P}(m_2)-{\mathfrak P}(m_1\wedge m_2)
\end{eqnarray}
is projector to the space $T(m_1,m_2)$, and the
\begin{eqnarray}\label{50Z}
{\mathfrak S}(m_1,m_2)&=&{\mathfrak P}(m_1\vee m_2)-{\mathfrak T}(m_1,m_2)
\nonumber\\&=&{\mathfrak P}(m_1\vee m_2)-{\mathfrak P}(m_1)-{\mathfrak P}(m_2)
+{\mathfrak P}(m_1\wedge m_2),
\end{eqnarray}
is projector to the space $S(m_1,m_2)$.
The dimension of the space $S(m_1,m_2)$ is $m_1\vee m_2-m_1-m_2+m_1\wedge m_2$.
In the special case that $m_1,m_2$ belong to the same chain, the dimension of the space $S(m_1,m_2)$ is $0$.
\end{itemize}
\end{proposition}
\begin{remark}\label{vb}
The ${\rm span}[H(m_1)\cup  H(m_2)]$ contains superpositions of states $\ket{X_{m_i}; a_i}$ where $a_i\in {\mathbb Z}(m_i)$
and $i=1,2$, which when embedded into $H(m_1\vee m_2)$, become $\ket{X_{m_1\vee m_2}; d_ia_i}$
where $d_i=(m_1\vee m_2)/m_i$ (Eq.(\ref{56})).
The disjunction (`merger') $\Sigma (m_1)\vee \Sigma (m_2)$, of two subsystems $\Sigma (m_1)$ and $\Sigma (m_2)$,
is described with the space $H(m_1\vee m_2)$ which is larger than the ${\rm span}[H(m_1)\cup  H(m_2)]$, and it contains 
additional states $\ket{X_{m_1\vee m_2}; a}$ with $a\in {\mathbb Z}^*(m_1\vee m_2)$, which span the space $S(m_1,m_2)$.
This is related to the fact that the group ${\mathbb Z}(m_1 \vee m_2)$ contains the elements of both 
${\mathbb Z}(m_1)$ and ${\mathbb Z}(m_2)$, plus the elements of ${\mathbb Z}^*(m_1\vee m_2)$, as explained in section \ref{AAS}.
Later, we will see the link between the space $S(m_1,m_2)$, with the supermodularity of $\ell (m|\rho _n)$ in Eq.(\ref{pro1}).
\end{remark}

We consider the system ${\Sigma}(n)$ in a state described with the density matrix $\rho_n$, and define the
\begin{eqnarray}\label{meas}
\ell (m|\rho _n)={\rm Tr}[\rho_n{\mathfrak P}(m)];\;\;\;\;\;\sigma (m_1,m_2|\rho _n)={\rm Tr}[\rho _n{\mathfrak S}(m_1,m_2)]
;\;\;\;\;\;m,m_1,m_2\in {\mathbb D}(n),
\end{eqnarray}
We also exclude the lowest state from ${\mathfrak P}(m)$ and define the 
\begin{eqnarray}
{\widetilde {\mathfrak P}}(m)={\mathfrak P}(m)-{\mathfrak P}(1);\;\;\;\;\;
{\widetilde \ell} (m|\rho _n)={\rm Tr}[\rho_n{\widetilde {\mathfrak P}}(m)].
\end{eqnarray}
We will use the notation ${\widetilde \Sigma} (m)$, for the subsystem $\Sigma (m)$ when we calculate probabilities excluding the lowest state.
With this notation, the probabilities in
$\Sigma (m)$ and ${\widetilde \Sigma} (\neg m)$ contain complementary information.
\begin{lemma}\label{37}
If $m\prec k$ then $\ell (m|\rho _n)\le \ell (k|\rho _n)$.
\end{lemma}
\begin{proof}
From Eq.(\ref{61}), it follows that ${\mathfrak P}(m){\mathfrak P}(k)={\mathfrak P}(m\wedge k)={\mathfrak P}(m)$.
Therefore ${\mathfrak P}(k)-{\mathfrak P}(m)$ is a projector orthogonal to ${\mathfrak P}(m)$ and
\begin{eqnarray}
\ell (k|\rho _n)=\ell (m|\rho _n)+{\rm Tr}\{\rho _n[{\mathfrak P}(k)-{\mathfrak P}(m)]\}
\end{eqnarray}
where ${\rm Tr}\{\rho _n[{\mathfrak P}(k)-{\mathfrak P}(m)]\}$ is a non-negative number. 
This proves the lemma.
\end{proof}
\begin{proposition}\label{propo}
\mbox{}
\begin{itemize}
\item[(1)]
The $\ell (m|\rho _n)$  obey the relation
\begin{eqnarray}\label{560}
\ell(m_1\vee m_2|\rho _n)-\ell (m_1|\rho _n)-\ell (m_2|\rho _n)
+\ell (m_1\wedge m_2|\rho _n)=\sigma (m_1,m_2|\rho _n).
\end{eqnarray}
From this follows that they obey the supermodularity inequality of Eq.(\ref{pro1}).
\item[(2)]
\begin{eqnarray}\label{89}
\ell (m|\rho _n)+{\widetilde \ell} (\neg m|\rho _n)\le \ell (\neg \neg m|\rho _n)+{\widetilde \ell} (\neg m|\rho _n)\le 1.
\end{eqnarray}
\end{itemize}
\end{proposition}
\begin{proof}
\mbox{}
\begin{itemize}
\item[(1)]
This follows immediately from Eq.(\ref{50Z}).
\item[(2)]
In the special case that $m_1=\neg m$ and $m_2=\neg \neg m$, Eq.(\ref{pro1}) reduces to
\begin{eqnarray}\label{560a}
1-\ell (\neg m|\rho _n)-\ell (\neg \neg m|\rho _n)
+\ell (1|\rho _n)\ge 0.
\end{eqnarray}
This proves one part of the inequality.
The second part, follows immediately from lemma \ref{37}, because $m\prec \neg \neg m$. 
\end{itemize}
\end{proof}

The supermodularity of $\ell(m|\rho _n)$ in Eq.(\ref{pro1}), is related to the term $\sigma (m_1,m_2|\rho _n)$
in Eq.(\ref{560}), which is related to the space $S(m_1,m_2)$, and also
to the fact that the group ${\mathbb Z}(m_1 \vee m_2)$ contains not only the elements of 
${\mathbb Z}(m_1)$ and ${\mathbb Z}(m_2)$, but also the elements of ${\mathbb Z}^*(m_1\vee m_2)$.
Within a chain $\sigma (m_1,m_2|\rho _n)=0$ for all density matrices, and the $\ell (m|\rho _n)$ 
obey Eq.(\ref{pro2}) like Kolmogorov probabilities.

\subsection{Example}\label{ex}
We consider the $\Lambda ({\bf \Sigma}_{18})$ which comprises of the subsystems of $\Sigma (18)$.
The projectors to these subsystems are
\begin{eqnarray}
&&{\mathfrak P}(1)=\ket{X_{18};0}\bra{X_{18};0}\nonumber\\
&&{\mathfrak P}(2)=\ket{X_{18};0}\bra{X_{18};0}+\ket{X_{18};9}\bra{X_{18};9}\nonumber\\
&&{\mathfrak P}(3)=\ket{X_{18};0}\bra{X_{18};0}+\ket{X_{18};6}\bra{X_{18};6}+\ket{X_{18};12}\bra{X_{18};12}\nonumber\\
&&{\mathfrak P}(6)=\sum _{\nu=0}^5\ket{X_{18};3\nu}\bra{X_{18};3\nu}\nonumber\\
&&{\mathfrak P}(9)=\sum _{\nu=0}^8\ket{X_{18};2\nu}\bra{X_{18};2\nu}\nonumber\\
&&{\mathfrak P}(18)={\bf 1}
\end{eqnarray}
In this case we have $3$ maximal contexts:
\begin{eqnarray}
{\bf \Sigma}_{18}(1)=\{\Sigma (18),\Sigma (9),\Sigma (3),\Sigma (1)\}\nonumber\\
{\bf \Sigma}_{18}(2)=\{\Sigma (18),\Sigma (6),\Sigma (3),\Sigma (1)\}\nonumber\\
{\bf \Sigma}_{18}(3)=\{\Sigma (18),\Sigma (6),\Sigma (2),\Sigma (1)\}
\end{eqnarray}
In $\Sigma (18)$ we consider the state
\begin{eqnarray}
\rho=\sum _{\nu=0}^{17}a_{\nu}\ket{X_{18};\nu}\bra{X_{18};\nu};\;\;\;\;\sum _{\nu=0}^{17}a_{\nu}=1;\;\;\;\;\;0\le a_{\nu} \le1,
\end{eqnarray}
We intentionally choose a mixed state with no off-diagonal elements, in order to emphasize that our arguments are not 
related to off-diagonal elements.
In table \ref{362} we give the $\ell(m|\rho )$ for all $m\in {\mathbb D}(18)$ (and also the $u(m|\rho )$ which is introduced later).

We next calculate the $\sigma (m_1,m_2)$ of Eq.(\ref{560}).
We take into into account the easily proved properties that $\sigma (1,m|\rho)=\sigma (18,m|\rho)=0$, and that if $m_1\prec m_2$ then
$\sigma (m_1,m_2|\rho)=0$. We find that 
\begin{eqnarray}
\sigma (9,6|\rho )&=&\ell(18|\rho)-\ell(9|\rho)-\ell(6|\rho)+\ell(3|\rho)\nonumber\\
&=&a_1+a_5+a_7+a_{11}+a_{13}+a_{17};\nonumber\\
\sigma (9,2|\rho)&=&\ell(18|\rho)-\ell(9|\rho)-\ell(2|\rho)+\ell(1|\rho)\nonumber\\
&=&a_1+a_3+a_5+a_{7}+a_{11}+a_{13}+a_{15}+a_{17};\nonumber\\
\sigma (2,3|\rho)&=&\ell(6|\rho)-\ell(2|\rho)-\ell(3|\rho)+\ell(1|\rho)=a_3+a_{15},
\end{eqnarray}
and that the rest $\sigma (m_1,m_2|\rho)=0$.
These values show that the $\ell (m|\rho)$ is a supermodular function (and therefore a superadditive function).

\section{Lower and upper quantum probabilities}

\subsection{The statement `belongs in $\Sigma (m)$' is not equivalent to 
`does not belong in $\neg {\widetilde \Sigma} (m)$'}

Comparison of Eqs.(\ref{76}),(\ref{1z}), for lower probabilities, with proposition \ref{propo} 
for quantum probabilities, shows that the  
$\ell (m|\rho _n)$ where $m\in {\mathbb D}(n)$,
has all the characteristics of a lower probability in the Dempster-Shafer theory. 
The upper probability is given by
\begin{eqnarray}\label{567}
u(m|\rho _n)=1-{\widetilde \ell} (\neg m|\rho _n)=1-{\ell}(\neg m|\rho _n)+{\ell}(1|\rho _n).
\end{eqnarray}
The ${\ell}(1|\rho _n)$ is added on the right hand side, so that $u(1|\rho _n)={\ell}(1|\rho _n)$.
For later use we also define the
\begin{eqnarray}
{\widetilde u}(m|\rho _n)=1-{\ell} (\neg m|\rho _n)=u(m|\rho _n)-{\ell}(1|\rho _n).
\end{eqnarray}
for which ${\widetilde u}(1|\rho _n)=0$.

Both $\ell (m|\rho _n)$ and $u(m|\rho _n)$ can be measured with the von Neumann measurement
\begin{eqnarray}\label{meas}
Q=\sum _{r=0}^{n-1} a_r\ket{X_n;r}\bra{X_n;r}.
\end{eqnarray}
We perform this measurement on many systems in the state $\rho _n$, and we count the number of times ${\mathfrak n}_r$ 
that the system will collapse into the state $\ket{X_n;r}$. Then
\begin{eqnarray}\label{41}
&&\ell (m|\rho _n)=\lim _{{\mathfrak n}_T\rightarrow \infty}
\frac{1}{{\mathfrak n}_T}\sum _{r\in {\cal L}}{\mathfrak n}_r;\;\;\;{\cal L}=\left \{0,\frac{n}{m},...,(m-1)\frac{n}{m}\right \}
\subseteq {\cal U}\nonumber\\
&&u (m|\rho _n)=\lim _{{\mathfrak n}_T\rightarrow \infty}
\frac{1}{{\mathfrak n}_T}\sum _{r\in {\cal U}}{\mathfrak n}_r;\;\;\;{\cal U}={\mathbb Z}(n)-\left \{\frac{n}{\neg m},...,(\neg m-1)\frac{n}{\neg m}
\right \}\nonumber\\
&&{\mathfrak n}_T=\sum _{s=0}^{n-1}{\mathfrak n}_s
\end{eqnarray}
In $\ell (m|\rho _n)$ we take $r\in {\cal L}$, which means that the collapsed state belongs entirely in $\Sigma (m)$ (as embedded into $\Sigma (n)$).
In $u (m|\rho _n)$ we take $r\in {\cal U}$, which means that the collapsed state does not belong in 
$\neg {\widetilde \Sigma} (m)={\widetilde \Sigma} (\neg m)$.
The statement `belongs in $\Sigma (m)$' is different from the statement `does not belong in $\neg {\widetilde \Sigma} (m)$',
and this is the reason for introducing lower and upper probabilities.
In contrast, in Kolmogorov's probability defined on the Boolean algebra associated with a powerset $2^{\Omega}$,
the statement `belongs to $A\subseteq \Omega$' is equivalent to the statement `does not belong to ${\overline A}=\Omega -A$'
(i.e., $q(A)=1-q({\overline A})$).

The difference between upper and lower probabilities is 
\begin{eqnarray}
&&d(m|\rho _n)=u(m|\rho _n)-\ell (m|\rho _n)={\rm Tr}[\rho_n{\mathfrak D}(m)];\;\;\;\;\;m\in {\mathbb D}(n)\nonumber\\
&&{\mathfrak D}(m)={\bf 1}_n-{\mathfrak P}(m)-{\widetilde {\mathfrak P}}(\neg m)={\bf 1}_n+{\mathfrak S}(m,\neg m)-{\mathfrak P}(m\vee \neg m)\nonumber\\
&&{\mathfrak D}(m){\mathfrak P}(m)=0;\;\;\;\;\;{\mathfrak D}(m)[{\bf 1}_n-{\mathfrak P}(m)]={\mathfrak D}(m).
\end{eqnarray}
In the Dempster terminology, $d(m|\rho _n)$ and ${\mathfrak D}(m)$ could be called `don't know' probability and
`don't know' projector, correspondingly.

The $d(m|\rho _n)$ can be calculated from the outcomes of the von Neumann measurement of Eq.(\ref{meas}), as follows:
\begin{eqnarray}
u (m|\rho _n)=\lim _{{\mathfrak n}_T\rightarrow \infty}
\frac{1}{{\mathfrak n}_T}\sum _{r}{\mathfrak n}_r;\;\;\;\;\;r\in {\cal U}-{\cal L}
\end{eqnarray}

\paragraph*{Ambiguity and multivaluedness:}
There are many probabilities between $\ell (m|\rho _n)$ and $u (m|\rho _n)$ which can be calculated using the 
outcomes ${\mathfrak n}_r$ from the von Neumann measurement of Eq.(\ref{meas}).
For example, in Eq.(\ref{41}) we can use $r\in S$ where
\begin{eqnarray}
S=\{0,\frac{n}{k},...,(k-1)\frac{n}{k}\};\;\;\;\;\;m\prec k\prec \neg \neg m;\;\;\;\;\;{\cal L}\subseteq S\subseteq {\cal U}
\end{eqnarray}
All these measurements show the `Dempster multivaluedness' \cite{DS1} in the present formalism.
For each subsystem $\Sigma (m)$ we have an interval of probabilities $[\ell (m|\rho _n), u (m|\rho _n)]$
which shows the existence of ambiguity.
This is an extra level of incertitude which is different from the 
usual uncertainties of non-commuting variables.

\subsection{Lower and upper quantum probabilities as capacities in ${\mathbb D}(n)$}

\begin{proposition}
Let $A=\{m_1,...,m_r\}\subseteq {\mathbb D}(n)$.
The lower and upper probabilities  with
\begin{eqnarray}
&&\ell(A|\rho _n)=\ell (m_1\vee...\vee m_r|\rho _n)\nonumber\\
&&u(A|\rho _n)=u(m_1\vee...\vee m_r|\rho _n)\nonumber\\
&&\ell(\emptyset|\rho _n)=u(\emptyset|\rho _n)=0
\end{eqnarray}
are capacities in ${\mathbb D}(n)$.
\end{proposition}
\begin{proof}
Let $m_A,m_B$ be the disjunctions (least common multipliers) of all elements in the sets $A\subseteq B\subseteq {\mathbb D}(n)$.
Then $m_A\prec m_B$ and therefore ${\ell}(m_A|\rho _n)\le {\ell}(m_B|\rho _n)$ (lemma \ref{37}).
This completes the proof for lower probabilities.

If $m\prec k$, then  $\neg k\prec \neg m$, and therefore $u(m|\rho _n)\le u(k|\rho _n)$.
This shows that lemma \ref{37} holds for upper probabilities, also.
Then the proof of the proposition for upper probabilities, is similar to the one above for lower probabilities.
\end{proof}
In analogy to Eq.(\ref{C}), we introduce the quantity
\begin{eqnarray}
{\mathfrak L}(m;k|\rho_n)=\ell(m \vee k|\rho _n)-\ell(k|\rho _n)
\end{eqnarray}
This quantifies the `added value' to the subsystem $\Sigma (m)$, if it
combines with the system $\Sigma (k)$, into the larger system $\Sigma (m\vee k)$ (see remark \ref{vb}).
In the case of coprime $k,m$, the ${\mathfrak L}(k;m|\rho_n)-{\widetilde \ell}(m|\rho_n)$ is a measure of the 
non-additivity of the lower probabilities.

If $\Sigma (m)$ is a subsystem of $\Sigma (k)$ (i.e., $m\prec k$), then ${\mathfrak L}(m;k|\rho_n)=0$.
In this case adding $\Sigma (m)$ to $\Sigma (k)$ does not have any effect, because $\Sigma (m)$ is already a part of $\Sigma (k)$. 
Also
\begin{eqnarray}
&&{\mathfrak L}(m;k\wedge m|\rho_n)=\ell(m|\rho _n)-\ell(k\wedge m|\rho _n)\nonumber\\
&&{\mathfrak L}(m;k|\rho_n)+{\mathfrak L}(m;k\wedge m|\rho _n)=\ell(m\vee k|\rho _n)-\ell (m|\rho _n)-\ell (k|\rho _n)
+\ell(m\wedge k|\rho _n).
\end{eqnarray}
Therefore the ${\mathfrak L}(m;k|\rho_n)+{\mathfrak L}(m;k\wedge m|\rho _n)$ is a measure of the deviation from the modularity property of Eq.(\ref{pro2}).

\paragraph*{Example:} We consider the example discussed earlier in section \ref{ex}.
In table \ref{362} we give the lower and upper probabilities $\ell(m|\rho )$ and $u(m|\rho )$ for all $m\in {\mathbb D}(18)$.
Using these values we calculate as an example, the ${\mathfrak L}(2;3|\rho)$ and the ${\mathfrak L}(2;3|\rho)-{\widetilde \ell}(2|\rho)$. We find
\begin{eqnarray}
&&{\mathfrak L}(2;3|\rho)=\ell (6|\rho)-\ell (3|\rho)=a_3+a_9+a_{15}\nonumber\\
&&{\mathfrak L}(2;3|\rho)-{\widetilde \ell}(2|\rho)=a_3+a_{15}.
\end{eqnarray}
The ${\mathfrak L}(2;3|\rho)-{\widetilde \ell}(2|\rho)$ is an example of the non-additive nature of the probabilities ${\ell}(m|\rho)$.

\subsection{Properties of the lower and upper quantum probabilities}

We first point out that
\begin{eqnarray}
u(1|\rho _n)={\ell}(1|\rho _n);\;\;\;\;\;u(n|\rho _n)={\ell}(n|\rho _n)=1.
\end{eqnarray}
We next introduce the
\begin{eqnarray}
&&{\overline {\ell}}(m|\rho _n)={\ell}(\neg m|\rho _n)-{\ell}(1|\rho _n)={\widetilde \ell} (\neg m|\rho _n)\nonumber\\
&&{\overline u}(m|\rho _n)=u(\neg m|\rho _n)-u(1|\rho _n)=1-\ell (\neg \neg m)
\end{eqnarray}
They are the analogues of $\ell (\overline A)$ and $u(\overline A)$ in section \ref{A}.

\begin{proposition}\label{propos}
\mbox{}
\begin{itemize}
\item[(1)]
The upper probabilities $u(m|\rho _n)$  obey the relation
\begin{eqnarray}\label{520}
u(m_1\vee m_2|\rho _n)-u(m_1|\rho _n)-u(m_2|\rho _n)
+u(m_1\wedge m_2|\rho _n)=-\sigma (\neg m_1,\neg m_2|\rho _n).
\end{eqnarray}

\item[(2)]
\begin{eqnarray}\label{P1}
\ell(m|\rho _n)+{\overline \ell}(m|\rho _n)\le 1 \le u(m|\rho _n)+{\overline u}(m|\rho _n)
\end{eqnarray}
\item[(3)]
\begin{eqnarray}\label{P2}
&&u(m|\rho _n)=u (\neg \neg m|\rho _n)\nonumber\\
&&\ell (\neg \neg m|\rho _n)-\ell (m|\rho _n)=u(m|\rho _n)-u (\neg m|\rho _n)\ge 0.
\end{eqnarray}
\item[(4)]
If $m\prec k\prec \neg \neg m$  then $\ell (m) \le \ell (k) \le u(m)$.
\item[(5)]
If
\begin{eqnarray}\label{bnm}
n=p_1^{e_1(n)}...p_N^{e_N(n)};\;\;\;\;m=p_1^{e_1(m)}...p_N^{e_N(m)};\;\;\;\;1\le e_i(m)\le e_i(n)
\end{eqnarray}
them $u(m)=1$
\end{itemize}
\end{proposition}
\begin{proof}
\mbox{}
\begin{itemize}
\item[(1)]
This follows from Eqs(\ref{560}), (\ref{567}).
\item[(2)]
This is proved using Eq.(\ref{89}).
\item[(3)]
Using Eq.(\ref{567}), we get
\begin{eqnarray}
u(\neg \neg m|\rho _n)=1-{\ell} (\neg \neg \neg m|\rho _n)+\ell(1|\rho _n)=1-{\ell} (\neg m|\rho _n)+\ell(1|\rho _n)=u(m|\rho _n).
\end{eqnarray}
Also 
\begin{eqnarray}
u(m|\rho _n)-u (\neg m|\rho _n)&=&[1-{\ell} (\neg m|\rho _n)+\ell(1|\rho _n)]-[1-{\ell} (\neg \neg m|\rho _n)+\ell(1|\rho _n)]\nonumber\\&=&
\ell (\neg \neg m|\rho _n)-\ell (m|\rho _n)
\end{eqnarray}
The right hand side is non-negative according to proposition \ref{37} and the fact that $m\prec \neg \neg m$.
\item[(4)]
Eq.(\ref{pro1}) with $m_1=\neg m$ and $m_2=k$ gives
\begin{eqnarray}\label{C1}
\ell(\neg m|\rho _n)+\ell(k|\rho _n)\le \ell(\neg m \vee k|\rho _n)+\ell(\neg m\wedge k|\rho _n)
\end{eqnarray}
But from $k\prec \neg\neg m$ it follows that $k\wedge m\prec \neg\neg m\wedge m=1$ and therefore $\ell(\neg m\wedge k|\rho _n)=\ell(1|\rho _n)$.
We rewrite Eq.(\ref{C1}) as
 \begin{eqnarray}
\ell(\neg m|\rho _n)+\ell(k|\rho _n)\le \ell(\neg m \vee k|\rho _n)+\ell(1|\rho _n)\le 1+\ell(1|\rho _n)
\end{eqnarray}
and from this follows that $\ell (k) \le u(m)$. Also, since $m\prec k$ we get $\ell (m) \le \ell (k)$ (proposition \ref{37}).
\item[(5)]
From Eq.(\ref{bnm}), it follows that $\neg m=1$ and therefore $u(m)=1$.
\end{itemize}
\end{proof}
\begin{remark}
\mbox{}
\begin{itemize}
\item[(1)]
The lower and upper probabilities in section \ref{A}, are defined on a Boolean algebra,  and therefore
${\overline {\overline A}}=A$, which implies that $\ell ({\overline {\overline A}})=\ell (A)$ and $u({\overline {\overline A}})=u(A)$. 
The analogue of this in our case which is a Heyting algebra, is Eq.(\ref{P2}).
\item[(2)]
From Eq.(\ref{520}), it follows that the upper probabilities obey the submodularity inequality 
\begin{eqnarray}
u(m_1\vee m_2|\rho _n)-u(m_1|\rho _n)-u(m_2|\rho _n)
+u(m_1\wedge m_2|\rho _n)\le 0.
\end{eqnarray}
\end{itemize}
\end{remark}

\section{Discussion}

We have considered the distributive lattice $\Lambda ({\bf \Sigma} _n)$ of subsystems of $\Sigma (n)$.
We have shown that with respect to this lattice, the lower and upper probabilities of 
the Dempster-Shafer approach, describe very well the 
quantum probabilities $\ell(m|\rho _n)$, for the following reasons:
\begin{itemize}
\item
For Kolmogorov probabilities $q(A)=1-q(\overline A)$ (Eq.(\ref{21})), but for
lower probabilities $\ell (A)$ is different from the $1-\ell (\overline A)$ (Eq.(\ref{1z})). 
The latter fits with the fact that in quantum systems 
`belongs in $\Sigma (m)$' is not the same as `does not belong in $\neg {\widetilde \Sigma} (m)$'.

\item
Kolmogorov probabilities satisfy the modularity equality of Eq.(\ref{pro2}),
but lower probabilities satisfy the supermodularity inequality of Eq.(\ref{76}).
The latter fits with the fact that the
quantum probabilities $\ell(m|\rho _n)$ satisfy the supermodularity inequality of Eq.(\ref{pro1}).

\item
There is multivaluedness and ambiguity (extra level of uncertainty, beyond the one associated with non-commuting variables) 
in quantum mechanics. The Dempster-Shafer theory is designed to describe similar situations in the classical world, and in this paper we applied it to the quantum world.
\end{itemize}
There is a long history of attempts to use more general (than Kolmogorov)
probabilistic theories in quantum mechanics\cite{G00,G0,G1,G2,G3,G4,G5}.
Operational approaches and convex geometry methods have been studied in \cite{CO1,CO2,CO3,CO4,CO5}.
Fuzzy phase spaces have been studied in \cite{FF1,FF2,FF3}.
Test spaces have been studied in \cite{TE1,TE2}.
Category theory methods have been studied in \cite{CA1,CA2}.
Topos theory methods have been used in \cite{TOP1,TOP2}.
In this paper we have used the Dempster-Shafer theory, for quantum probabilities in the distributive lattice of subsystems.
The Dempster-Shafer theory, for quantum probabilities in the Birkhoff-von Neumann orthomodular lattice of subspaces will be discussed 
elsewhere\cite{vv}

\newpage

\begin{table}
\caption{The project marks for four students $S_1$, $S_2$, $S_3$, $S_4$}
\centering
\begin{tabular}{|c||c|c|c|c|}\hline
$S_1$& $60$& $65$& $72$&\\ \hline
$S_2$& $70$& $72$& &\\ \hline
$S_3$& $61$& $65$& $68$&\\ \hline
$S_4$& $50$& $55$& $58$&$62$\\ \hline
\end{tabular}\label{360}

\vspace{0.7cm}

\caption{The lower and upper probabilities corresponding to the sets $A_1$, $A_2$, $A_1\cup A_2$, $A_3$ of Eq.(\ref{79})}
\centering
\begin{tabular}{|c||c|c|c|c|}\hline
&$A_1$& $A_2$& $A_1\cup A_2$& $A_3$\\ \hline\hline
$\ell (A_i)$& $1/4$& $1/4$& $3/4$&$1/4$\\ \hline
$u(A_i)$& $3/4$& $1/2$& $1$&$1$\\ \hline
\end{tabular}\label{361}

\vspace{0.7cm}

\caption{The lower and upper probabilities for example \ref{ex}}
\centering
\begin{tabular}{|c||c|c|}\hline
$m$& $\ell (m|\rho )$& $u(m|\rho )$\\ \hline\hline
$1$& $a_0$& $a_0$ \\ \hline
$2$& $a_0+a_9$& $a_0+\sum _{\nu=0}^8 a_{2\nu+1}$ \\ \hline
$3$& $a_0+a_6+a_{12}$& $\sum a_{\nu};\;\;\nu\ne 9$ \\ \hline
$6$& $\sum _{\nu=0}^5 a_{3\nu}$& $1$ \\ \hline
$9$&  $\sum _{\nu=0}^8 a_{2\nu}$& $\sum a_{\nu} ;\;\;\nu\ne 9$\\ \hline
$18$& $1$& $1$ \\ \hline
\end{tabular}\label{362}

\end{table}

\end{document}